\documentclass[letterpaper, 10 pt, conference]{ieeeconf}  

\IEEEoverridecommandlockouts                              
\usepackage{amsmath}
\usepackage{mathtools}
\usepackage{amssymb}
\usepackage{graphicx}
\usepackage{cancel}
\usepackage{tikz}
\usetikzlibrary{angles}
\newtheorem{theorem}{Theorem}

\newtheorem{example}{Example}
\newtheorem{remark}{Remark}

\newcommand\xqed[1]{%
  \leavevmode\unskip\penalty9999 \hbox{}\nobreak\hfill
  \quad\hbox{#1}}
\newcommand\exampleend{\xqed{$\triangle$}}

\usepackage{booktabs}

\newcommand{\real}{\mathbb{R}}
\newcommand{\E}{\mathbb{E}}
\newcommand{\tr}{\mathrm{tr}}
\newcommand{\re}{\mathrm{Re}}
\newcommand{\im}{\mathrm{Im}}
\newcommand{\diag}{\mathrm{diag}}

\title{\LARGE \bf
Dynamic Brain Networks with Prescribed Functional Connectivity
}

\author{U. Casti, G. Baggio, D. Benozzo, S. Zampieri, A. Bertoldo, A. Chiuso
\thanks{The authors are with the Department of Information Engineering, University of Padova, Italy. This research was supported by the DEI Proactive grant ``Personalized whole brain models for neuroscience: inference and validation'' from the Department of Information Engineering, University of Padova. }%
}

\begin{document}

\maketitle
\thispagestyle{empty}
\pagestyle{empty}

\begin{abstract}
In this paper, we consider stable stochastic linear systems modeling whole-brain resting-state dynamics. We parametrize the state matrix of the system (effective connectivity) in terms of its steady-state covariance matrix (functional connectivity) and a skew-symmetric matrix $S$.  We examine how the matrix $S$ influences some relevant dynamic properties of the system. Specifically, we show that a large $S$ enhances the degree of stability and excitability of the system, and makes the latter more responsive to~high-frequency~inputs. 

\end{abstract}

\section{Introduction}

Various network models have been proposed in neuroscience for describing the brain organization and activity \cite{bullmore2009complex,friston2011functional}. A commonly used one is the Structural Connectivity (SC) model, which is based on brain anatomy. This model builds a graph whose nodes represent brain regions (i.e., populations of strongly interconnected neurons) and edges anatomical connections between them. 
Although the SC model proved to be an extremely useful description of the brain, its interpretative power is~limited~by~its~static~nature. 

Modern noninvasive imaging techniques, such as functional magnetic resonance imaging (fMRI), electroencephalography (EEG), and magnetoencephalography (MEG) generate data, in the form of time series, describing the neural activity of brain regions \cite{bowman2014brain}. 
It can be seen that the activity of some of them appear to be correlated and hence we can build a graph linking regions that exhibit such a correlation. 
The network obtained in this way is called Functional Connectivity (FC), since this coupling can be interpreted as a functional cooperation between regions. Notice that the correlation between regions changes over time according to the task the brain is accomplishing. For this reason, much research on this field has focused on the resting-state Functional Connectivity (rsFC), that is obtained from measurements recorded in the absence of stimuli or tasks \cite{bijsterbosch2017introduction}.
The degree of correlation between regions can be mathematically specified in different ways. One simple and natural way is by means of the statistical correlation between the brain activity signals. In whatever way it is obtained, the FC matrix is always symmetric and hence unable to capture the causal~relationship~between~brain~regions.

A model that attempts to unveil these causal interactions is the so-called Effective Connectivity (EC) model. It is a generative model in the sense that it provides a mathematical description by which, in principle, it is possible to replicate the observed signals. Different dynamical models can be employed to this aim, the simplest one being a linear stochastic model known as linear Dynamic Causal Model (linear DCM) \cite{friston2003dynamic}. In this case, finding the EC consists in finding the linear system that best fits the observed data (which is the standard goal of system identification), e.g., see \cite{frassle2018generative,prando2020sparse,gindullina2021estimating}. In this setup, the EC model is specified by the state interaction matrix of the estimated linear system and the covariance matrix of its driving noise. 

Notice that, since an EC model can in principle simulate the brain activity, it can also predict the FC in the form of the (steady-state) covariance matrix associated with the linear stochastic system. However, the EC model contains more information than its associated FC.  How to extract this additional information and how to translate it in terms of macroscopic properties of the brain is currently an open problem in neuroscience. Motivated by this problem, in this paper we consider linear stochastic systems with fixed FC (i.e. with fixed steady-state covariance) and analyze how the dynamic properties of the system are affected by the remaining degrees of freedom of the model, which are shown to be encoded in the entries~of~a~skew-symmetric~matrix.

\noindent \textbf{Contribution.} The starting point of our analysis is a parametrization of the state matrix $A$ of a stable linear stochastic system in terms of a positive definite matrix $\Sigma$, representing its steady-state covariance, and a skew-symmetric matrix $S$. We then examine how the ``size'' of $S$ affects some dynamic properties of the system. Specifically, we show through analytical and numerical results that, as $S$ increases, (i) the stability margin of the system, as measured by the modulus of the largest real part of the eigenvalues of $A$, grows, (ii) the transient amplification (or excitability) of the system, as quantified by the numerical abscissa of $A$, increases, (iii) the spectral energy of the system shifts from low to high frequencies. Altogether, our findings suggest that $S$ may play a crucial role in shaping brain dynamics.

\noindent \textbf{Related work.} A few works in neuroscience tried to extract, interpret, and analyze information from EC that goes beyond functional relationships between brain regions. In particular, in \cite{friston2011network} the authors use the skew-symmetric part of the EC matrix to build a hierarchy between brain regions. More closely related to our work are \cite{lin2017differential,chen2022dynamical}, which propose the Differential Covariance (DC) as a tool to infer directionality of interactions between brain regions. The DC is a matrix whose entries describe the correlation between the activity of a brain area and the variation of the activity occurring in another one. As it will be clear later, the DC matrix is strongly related to the matrix $S$ investigated~in~this~paper.

\noindent \textbf{Notation.} Given $x\in\mathbb{C}$, $\re[x]$ and $\im[x]$ denote the real and imaginary part of $x$, respectively. The symbol $\imath$ stands for the imaginary unit. 
Given a matrix $A\in\mathbb{C}^{n\times m}$ we denote with $A^\top$ the transpose of $A$ and with $A^*$ the conjugate transpose of $A$. For an Hermitian matrix $A=A^*$, $\lambda_{\max}(A)$ and $\lambda_{\min}(A)$ denote the largest and smallest eigenvalue of $A$, respectively. A positive definite (semidefinite) matrix $A$ is denoted by $A\succ 0$ ($A\succeq 0$, respectively). We let $\|A\|$ denote the 2-norm of a matrix, $\tr(A)$ the trace of $A$, $I_n$ the $n$-dimensional identity matrix (the subscript $n$ will be dropped when clear from the context), and $\diag(d_1,\dots,d_n)$ the diagonal matrix with entries $d_1,\dots,d_n$ on the diagonal. We say that a matrix $A\in\mathbb{R}^{n\times n}$ is Hurwitz stable if all the eigenvalues of $A$ have strictly negative real part.


\section{Preliminaries}\label{eq:preliminaries}


We consider the continuous-time linear time-invariant stochastic system
\begin{align}\label{eq:model}
    \dot x(t) &= A x(t) + w(t),
\end{align}
where $x\in\mathbb{R}^n$ is the vector containing the states of the network nodes, $A\in\mathbb{R}^{n\times n}$ is the state interaction matrix, and $w$ is a zero-mean white noise vector with positive definite covariance matrix $\mathbb{E}[w(t)w(t)^\top]=:\Sigma_w\succ 0$.\footnote{Rigorosuly speaking, the equation \eqref{eq:model} should be intended in the differential form $\mathrm{d}x = A x\mathrm{d}t + \mathrm{d}w$, where $\mathrm{d}w$ is a Wiener process, i.e., a process with stationary orthogonal increments, e.g. see \cite[Chap. 3, Sec. 4]{astrom1970stochastic}.}

When modeling resting-state brain activity, \eqref{eq:model} is referred to as linear DCM \cite{friston2003dynamic}, in which $x$ contains the states of brain regions, $A$ is the EC matrix encoding causal relationships between brain regions, and $\Sigma_w = \sigma^2 I$, $\sigma\in\mathbb{R}$. The DCM model also contains an output equation of the form  
$y(t) = [h*x](t) + v(t),$
where $y\in\mathbb{R}^n$ is the Blood-Oxygen-Level-Dependent (BOLD) signal, $h\colon \mathbb{R}^{n}\to \mathbb{R}^n$ is the hemodynamic function mapping brain activity to the BOLD signal, the symbol $*$ stands for convolution, and $v$ is white measurement noise.  

In this paper, $A$ is assumed to be Hurwitz stable. This ensures the existence of a positive definite steady-state or stationary state covariance matrix:
\begin{align*}
    \Sigma &:= \lim_{t\to \infty} \mathbb{E}\left[x(t)x(t)^\top\right]\succ 0,
\end{align*}
which can be computed as the solution of the algebraic Lyapunov equation \cite[Thm.~6.1]{astrom1970stochastic}:
\begin{align}\label{eq:lyap}
    A\Sigma + \Sigma A^\top +\Sigma_w=0.
\end{align}

In linear DCM, $\Sigma$ can be seen as a measure of functional strength between brain regions since it is related to the steady-state covariance of the output which (up to a normalization) coincides with the FC. Therefore, $\Sigma$ can be thought of as a proxy of the FC.

\section{Parametrizing systems with prescribed $\Sigma$} \label{sec:parametrization}

Consider the system \eqref{eq:model} and define the sets
\begin{align}\label{eq:sets}
\begin{aligned}
\mathcal{A} &:= \{A\in\real^{n\times n} \,:\, A \text{ Hurwitz stable}\},\\
\mathcal{P} &:= \{P\in\real^{n\times n} \,:\, P=P^\top\succ 0\}, \\
\mathcal{S} &:= \{S\in\real^{n\times n} \,:\, S=-S^\top\}.
\end{aligned}
\end{align}
Our first result provides a parametrization of Hurwitz stable state matrices $A$ in terms of a pair of matrices, namely a positive definite matrix and a skew-symmetric one. 

\begin{theorem}{\em\bfseries(Parametrization of Hurwitz stable $A$'s)}\label{thm-param}
Consider the system \eqref{eq:model} and the sets defined in \eqref{eq:sets}. For any $\Sigma \in\mathcal{P}$ and $S\in\mathcal{S}$, 
\begin{align}\label{eq:param}
    A=\left(-\frac{1}{2} \Sigma_w + S\right)\Sigma^{-1} \in \mathcal{A},
\end{align}
and $\Sigma$ solves \eqref{eq:lyap}. Further, the map
\begin{align}\label{eq:bijection}
\begin{aligned}
    f\colon \mathcal{P}\times \mathcal{S} &\to \mathcal{A},\\
    \ (\Sigma,S) &\mapsto A=\left(-\frac{1}{2} \Sigma_w + S\right)\Sigma^{-1},
\end{aligned}
\end{align}
is a bijection.
\end{theorem}
\begin{proof}
By construction, the matrices $\Sigma$ and $A$ in \eqref{eq:param} satisfy the Lyapunov equation \eqref{eq:lyap}. 
Since $\Sigma\succ 0$ and $\Sigma_w\succ 0$, it follows that $A\in\mathcal{A}$ \cite[Thm. 8.2]{hespanha2018linear}. To prove that \eqref{eq:bijection} is a bijection observe that, for any $A\in\mathcal{A}$, there exists a unique $\Sigma\in\mathcal{P}$ satisfying \eqref{eq:lyap} \cite[Thm. 8.2]{hespanha2018linear}. Further, it holds
\begin{align}
    A\Sigma &= \frac{1}{2}\left(A\Sigma + \Sigma A^\top\right) + \frac{1}{2}\left(A\Sigma - \Sigma A^\top\right)\\
    &= -\frac{1}{2}\Sigma_w + S,\label{eq:proof-bijection}
\end{align}
where $S:=\frac{1}{2}\left(A\Sigma - \Sigma A^\top\right)\in\mathcal{S}$. Thus, for any $A\in\mathcal{A}$ there exists $(\Sigma,S)\in \mathcal{P}\times \mathcal{S}$ such that \eqref{eq:param} holds, that is, \eqref{eq:bijection} is surjective. Finally, consider $A_1,A_2\in\mathcal{A}$ and let $\Sigma_1, \Sigma_2\in \mathcal{P}$ be the corresponding solutions to \eqref{eq:lyap} and $S_1:=\frac{1}{2}\left(A_1\Sigma_1 - \Sigma_1 A_1^\top\right)$, $S_2:=\frac{1}{2}\left(A_2\Sigma_2 - \Sigma_2 A_2^\top\right)$. If $A_1=A_2$, then $\Sigma_1=\Sigma_2$ by the uniqueness of the solution to \eqref{eq:lyap} and, from \eqref{eq:proof-bijection}, $A_1\Sigma_1=A_2\Sigma_2$ implies $S_1=S_2$. This proves injectivity of \eqref{eq:bijection} and concludes the proof.
\end{proof}

As a consequence of the previous result, all state matrices of \eqref{eq:model} generating a fixed steady-state covariance $\Sigma\succ 0$ are parametrized by the skew-symmetric matrix $S$ of Equation \eqref{eq:param}. Further, for a fixed $A\in\mathcal{A}$ in \eqref{eq:model} it is possible to obtain the corresponding parameters $\Sigma$ and $S$ as follows: $\Sigma$ corresponds to the (unique) solution to \eqref{eq:lyap} and $S$ to the skew-symmetric part of $A\Sigma$, i.e., $S=\frac{1}{2}\left(A\Sigma-\Sigma A^\top\right)$.

\begin{remark}{\em\bfseries(Statistical interpretation of $S$)}\label{rmk:statistical}
The skew-symmetric matrix $S$ in \eqref{eq:bijection} can be also characterized from a statistical viewpoint. In fact, it can be shown that
\begin{align*}
    \lim_{t\to \infty}\mathbb{E}[\dot x(t) x(t)^\top] &= 
    A\Sigma = -\frac{1}{2} \Sigma_w + S, 
\end{align*}
The latter equation yields
\begin{align*}
    S = \frac{1}{2} \Sigma_w + \lim_{t\to \infty} \E[\dot x(t) x(t)^\top].
\end{align*}
In \cite{lin2017differential} the matrix $\lim_{t\to \infty} \E[\dot x(t) x(t)^\top]$ is termed Differential Covariance (DC).  Note in particular that when $\Sigma_w$ is diagonal (as in the DCM model) the off-diagonal entries of $S$ coincide with those of the DC.

\begin{example}{\em\bfseries(Matrices $S$ in EC mice data)}
We consider the experimental EC matrices employed in \cite{benozzo2023macroscale} which are estimated from BOLD recordings in mice. In Fig.~\ref{fig:box}, we plot the norm of the $S$ matrix constructed from the parametrization in Eq.~\eqref{eq:param} for a sample of $N=20$ subjects, and compare it with an $S$ extracted from randomly generated stable matrices (see the caption of Fig.~\ref{fig:box} for details). The mice ECs exhibit a norm of $S$ substantially larger than that of the random case.
In the next sections, we elucidate the role played by a ``large'' $S$ on~some~relevant~dynamic~properties~of~the~system.\exampleend
    \begin{figure}\label{fig:box}
        \centering
        \includegraphics[scale=0.8]{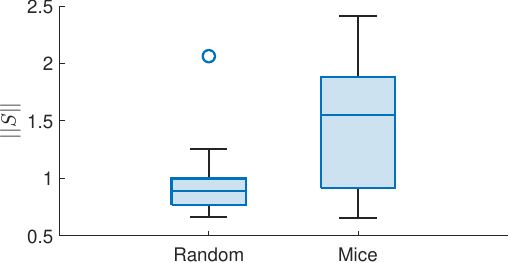}
        \caption{
        The figure shows the norm of $S$ in the parametrization \eqref{eq:param} for mice ECs matrices and randomly generated Hurwitz stable matrices. The entries of the random matrices are drawn independently from a standard normal distribution. To enable a fair comparison with the mice ECs, these matrices are shifted and rescaled so that they possess the same maximum real part (stablity margin) and imaginary part of the eigenvalues of mice ECs (averaged over the $N$ subjects). The box plot shows the results obtained for $20$ realizations of the random ensemble and $20$ mice subjects with $\Sigma_w=I$.\vspace{-0.25cm}}
        \label{fig:my_label}
    \end{figure}
\end{example}

\end{remark}

\section{Role of $S$ in the stability of \eqref{eq:model}}\label{sec:stability}

We consider the effect of large perturbations on the skew-symmetric term $S$ in \eqref{eq:param} on the stability of $A$. We first note that, for $S=0$, the eigenvalues of $A$ are real since $\Sigma_w\Sigma^{-1}$ is similar to a symmetric matrix. Further, the average of the eigenvalues is independent of $S$, namely
$$
\frac{1}{n}\tr\left(A\right) = -\frac{1}{2n}\tr\left(\Sigma_w\Sigma^{-1}\right),
$$ 
since $S\Sigma^{-1}$ is similar to a skew-symmetric matrix and therefore has trace equal to zero. When the entries of $S$ increase, the eigenvalues of $A$ become complex and their real part tend to get closer to the average $-\frac{1}{2n}\tr\left(\Sigma_w\Sigma^{-1}\right)$, which means that the system becomes more stable, {i.e., the largest real part of the eigenvalues of $A$ becomes more negative}. This is clearly seen in the 2-dimensional case,~as~discussed~below.

\begin{example}{\em\bfseries(2-dimensional~system)}
\label{ex:stab}
Consider the parametrization of $A\in\real^{2\times 2}$ given in Theorem \ref{thm-param} and let 
$$
S=\alpha \begin{bmatrix}0 & 1\\ -1 & 0\end{bmatrix}, \ \ \ \alpha\ge 0.
$$
We fix $\Sigma_w = \sigma^2 I$ and 
we take $\Sigma^{-1} = \diag(d_1,d_2)$ with $d_1\geq d_2 > 0$.
Then
$$
\label{eq:simA}
A = \begin{bmatrix} -\frac{\sigma^2}{2}d_1 & \alpha d_2 \\ -\alpha d_1 & -\frac{\sigma^2}{2}d_2\end{bmatrix}.
$$
The eigenvalues of $A$ are the roots of its characteristic polynomial
$$
p(\lambda) = \lambda^2 + \frac{\sigma^2}{2}(d_1+d_2)\lambda +\left(\frac{\sigma^4}{4}+\alpha^2\right)d_1d_2.
$$
It follows that the eigenvalues are real if $0\le\alpha\le\left(d_1-d_2\right)/\sqrt{d_1d_2}$, and complex conjugate with identical real part $-\frac{\sigma^2}{4}(d_1+d_2)$, otherwise. The behavior of the eigenvalues as a function of $\alpha\geq0$ is illustrated in Fig.~\ref{fig:rlocus}.~\exampleend
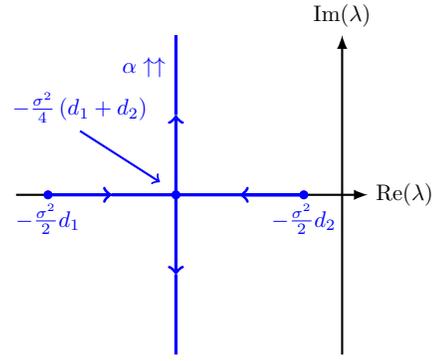
\begin{figure}[!h]
    \centering
    \begin{tikzpicture}[thick,scale=0.85, every node/.style={transform shape}]
    \draw[-latex] (0,0) -- (5.5,0) node[right] {$\re(\lambda)$}; 
    \draw[-latex] (5.1,-2.5) -- (5.1,2.5) node[above] {$\im(\lambda)$}; 
    
    \node[blue,scale=1] at (2,2) {$\alpha\uparrow\uparrow$};
    \node[inner sep=0.05cm,circle,fill=blue] at (0.5,0) {};
    \node[inner sep=0.05cm,circle,fill=blue] at (4.5,0) {};
    \node[inner sep=0.05cm,circle,fill=blue] at (2.5,0) {};
    \draw[very thick,blue,->] (0.5,0) node[below] {$-\frac{\sigma^2}{2}d_1$} -- (1.5,0);
    \draw[very thick,blue] (1.5,0) -- (2.5,0);
    \draw[very thick,blue,->] (4.5,0) node[below]{$-\frac{\sigma^2}{2}d_2$} -- (3.5,0);
    \draw[very thick,blue,] (4.5,0) -- (2.5,0) ;
    \draw[,blue,->] (1,1) node[above]{$-\frac{\sigma^2}{4}\left(d_1+d_2\right)$} -- (2.25,0.2);
    \draw[very thick,blue,->] (2.5,0) -- (2.5,1.25);
    \draw[very thick,blue] (2.5,1.25) -- (2.5,2.5);
    \draw[very thick,blue,->] (2.5,0) -- (2.5,-1.25);
    \draw[very thick,blue] (2.5,-1.25) -- (2.5,-2.5);
    \end{tikzpicture}
    \caption{Behavior of the eigenvalues of $A$ for the system in Example \ref{ex:stab} and increasing values of $\alpha\geq0$.\vspace{-0.25cm}}
    \label{fig:rlocus}
\end{figure}
\end{example}

The next result provides asymptotic formulas for the eigenvalues of $A$ when the entries of $S$ grow, in the general $n$-dimensional case.

\begin{theorem}{\em\bfseries (Behavior of eigenvalues of $A$ for large $S$)}
\label{thm:stab} Consider the parametrization of $A$ given in Theorem \ref{thm-param} and let $S=\alpha \bar{S}$, where $\alpha\in\real$ and $\bar{S}\in\mathcal{S}$, $\bar{S}\ne 0$. Let $\{\lambda_i\}_{i=1}^n$ be the eigenvalues of $A$ and $\mu_i$, $u_i$ be the (imaginary) eigenvalues and the (normalized) eigenvectors of $\Sigma^{-1/2}\bar{S}\Sigma^{-1/2}$. {Assume that the eigenvalues of $\Sigma^{-1/2}\bar{S}\Sigma^{-1/2}$ are simple.} Then,~as~$\alpha\to\infty$,
\begin{align}\label{eq:limit}
    &\re[\lambda_i] \to -\frac{1}{2}u_i^* \Sigma^{-1/2}\Sigma_w\Sigma^{-1/2} u_i,\\
    &\alpha^{-1}\im[\lambda_i] \to \im[\mu_i].\notag
\end{align}
\end{theorem}\bigskip
\begin{proof}
Observe that $A$ is similar to the matrix
\begin{align*}
    A'&=-\frac{1}{2}\Sigma^{-1/2}\Sigma_w\Sigma^{-1/2} + \alpha \Sigma^{-1/2}\bar{S}\Sigma^{-1/2}\\ &= \frac{1}{\varepsilon}\left(-\frac{\varepsilon}{2}\Sigma^{-1/2}\Sigma_w\Sigma^{-1/2} + \Sigma^{-1/2}\bar{S}\Sigma^{-1/2}\right),
\end{align*}
with $\varepsilon := \alpha^{-1}$.
Consider the matrix $-\frac{\varepsilon}{2}\Sigma^{-1/2}\Sigma_w\Sigma^{-1/2}+\Sigma^{-1/2}\bar{S}\Sigma^{-1/2}$ and let $\{\mu_i(\varepsilon)\}$ and $\{u_i(\varepsilon)\}$ be its eigenvalues and corresponding eigenvectors. 
{The latter eigenvalues and eigenvectors are analytic functions of $\varepsilon$ in a neighbourhood of $\varepsilon=0$, since the eigenvalues of $\Sigma^{-1/2}\bar{S}\Sigma^{-1/2}$ are simple by assumption \cite{magnus1985differentiating}, and we can consider their Taylor expansion around $\varepsilon=0$:}
$$
\mu_i(\varepsilon) = \mu_{i,0} + \mu_{i,1}\varepsilon+\cdots, \quad u_i(\varepsilon) = u_{i,0} + u_{i,1}\varepsilon+\cdots,
$$
where $\mu_{i,0}=\mu_i$ is imaginary. 
Then, 
\begin{align*}
    \left(-\frac{\varepsilon}{2}\Sigma^{-1/2}\Sigma_w\Sigma^{-1/2}+\Sigma^{-1/2}\bar{S}\Sigma^{-1/2}\right)&(u_{i,0}+u_{i,1}\varepsilon+\cdots) \\=  (\mu_{i,0}+\mu_{i,1}\varepsilon+\cdots)&(u_{i,0}+u_{i,1}\varepsilon+\cdots).
\end{align*}
By equating the terms of the same order in the previous expression, we obtain
\begin{align*}
    &\Sigma^{-1/2}\bar{S}\Sigma^{-1/2} u_{i,0} =\mu_{i,0} u_{i,0},& \\
    &\Sigma^{-1\!/2}\bar{S}\Sigma^{-1\!/2} u_{i,1}\! \!-\!\!\frac{1}{2}\Sigma^{-1\!/2}\Sigma_w\Sigma^{-1\!/2} u_{i,0} \!=\! \mu_{i,1} u_{i,0}\! +\! \mu_{i,0} u_{i,1}.
\end{align*}
The first equation is equivalent to $u_{i,0}^*\Sigma^{-1/2}\bar{S}\Sigma^{-1/2}=\mu_{i,0}u^*_{i,0}$, which substituted in the second one (premultiplied by $u_{i,0}^*$) yields
$
    -\frac{1}{2}u_{i,0}^*\Sigma^{-1/2}\Sigma_w\Sigma^{-1/2}u_{i,0}=\mu_{i,1} u_{i,0}^* u_{i,0}.
$
The latter identity in turn implies
\begin{align}    
    \mu_{i,1} &= -\frac{u_{i,0}^* \Sigma^{-1/2}\Sigma_w\Sigma^{-1/2} u_{i,0}}{2\|u_{i,0}\|^2}\notag\\
    &= -\frac{1}{2}u_{i,0}^* \Sigma^{-1/2}\Sigma_w\Sigma^{-1/2} u_{i,0}\in\real.\label{eq:proof-asymp}
\end{align}
Consider now the matrix $\frac{1}{\varepsilon}(-\frac{\varepsilon}{2}\Sigma^{-1/2}\Sigma_w\Sigma^{-1/2}+\Sigma^{-1/2}\bar{S}\Sigma^{-1/2})$. As $\varepsilon\to 0$ the eigenvalues of this matrix satisfy
$
    \frac{1}{\varepsilon}\mu_i(\varepsilon) \to \frac{1}{\varepsilon}(\mu_{i,0}+\varepsilon \mu_{i,1}) = \mu_{i,1} +\frac{1}{\varepsilon} \mu_{i,0}$.
~Then, 
$$\alpha^{-1}\im[\lambda_i] \to \im[\mu_{i,0}]=\im[\mu_i].$$
Moreover, since $\mu_{i,0}$ is purely imaginary, as $\varepsilon\to 0$, from \eqref{eq:proof-asymp} we have
\begin{align*}
    \re\left[\frac{1}{\varepsilon}\mu_i(\varepsilon)\right] &\to \re\left[\frac{1}{\varepsilon}(\mu_{i,0}+\varepsilon \mu_{i,1})\right] = \mu_{i,1} \\&= -\frac{1}{2}u_{i,0}^* \Sigma^{-1/2}\Sigma_w\Sigma^{-1/2} u_{i,0}.
\end{align*}
The thesis follows by noting that $\varepsilon\to 0$ is equivalent to $\alpha\to\infty$ and $\{u_{i,0}\}$ coincide with the (normalized)~eigenvectors~of~$\Sigma^{-1/2}\bar{S}\Sigma^{-1/2}$. 
\end{proof}

Theorem \ref{thm:stab} states that while the real parts of the eigenvalues of $A$ tends to a constant as the size of $S$ grows, their imaginary part diverge to infinity with rate given by the eigenvalues of $\Sigma^{-1/2}\bar{S}\Sigma^{-1/2}$.
It is also possible to prove that the limit values of the real parts fall within the range of the minimum and maximum eigenvalues of the matrix $-\frac{1}{2}\Sigma^{-1/2}\Sigma_w\Sigma^{-1/2}$. This implies that the real parts of the eigenvalues do not exceed the bounds given by the eigenvalues of the matrix $A$ associated with $\alpha = 0$. 

{Moreover, numerical simulations suggest that the real parts of the eigenvalues of $A$ tend to concentrate around the mean as $\alpha$ increases (see also Fig.~\ref{fig:ex1} for a 4-dimensional example).}
\begin{figure}[!h]
    \centering
    \vspace{-0.15cm}
    \includegraphics[width=0.8\linewidth]{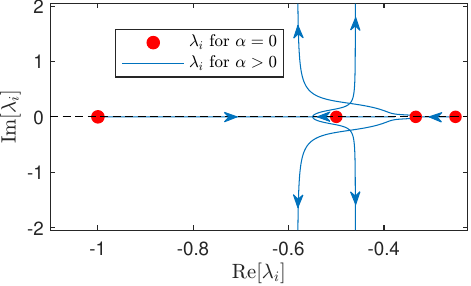}
    \caption{Behavior of the eigenvalues of $A$ for increasing values of $\alpha\geq0$ for a $4$-dimensional system the asymptotic branches corresponds to the values $-\frac{1}{2}u_i^* \Sigma^{-1/2}\Sigma_w\Sigma^{-1/2} u_i$ of Equation \eqref{eq:limit}.\vspace{-0.25cm}}
    \label{fig:ex1}
\end{figure}

\section{Role of $S$ in the excitability of \eqref{eq:model}}\label{sec:excitability}

While the previous analysis is strictly linked to the asymptotic behavior of the system \eqref{eq:model}, in this section we will focus on a different metric that characterizes the transient evolution of the system. 
This metric is the numerical abscissa of~$A$,~defined~as:
\begin{equation*}
    \omega(A) := \lambda_{\max}\left(\frac{A + A^\top}{2}\right).
\end{equation*}
The numerical abscissa quantifies the initial growth rate of the trajectories generated by \eqref{eq:model}. Precisely, it holds \cite{Trefethen1999}:
\begin{equation}\label{eq:omega}
    \frac{\mathrm{d}\| e^{At} \|}{\mathrm{d}t} \bigg|_{t = 0} = \omega(A).
\end{equation}
In view of \eqref{eq:omega}, we say that a system is excitable when $\omega(A)>0$, and non-excitable otherwise. Excitability has been shown to be beneficial for information transmission and controllability of linear dynamical systems \cite{baggio2021non,baggio2022energy}.

The following example provides insights into the relationship between the numerical abscissa of $A$ and the~size~of~$S$. 
\begin{example}{\em\bfseries(2-dimensional~system, cont'd)}\label{ex:numAbscissa}
Consider the system of Example \ref{ex:stab}. The numerical abscissa $\omega(A)$ can be expressed in closed form as:
\begin{equation*}
    \omega(A) = -\frac{\sigma^2}{4}(d_1+d_2) + \frac{\sqrt{\sigma^4+2\alpha^2}}{4}\left(d_1-d_2\right).
\end{equation*}
When $d_1=d_2$ the system is non-excitable for all $\alpha$. When $d_1\ne d_2$ it is excitable for $\alpha > \sigma^2\frac{\sqrt{d_1d_2}}{d_1-d_2}$. Further, as $\alpha$ increases, the numerical abscissa grows unbounded as a linear function of $\alpha$, namely $\omega(A)\approx \frac{d_1 - d_2 }{2\sqrt{2}}\alpha $.\exampleend
\begin{figure}
    \centering
    \begin{tikzpicture}
    \coordinate (O) at (0.5*8/1.75+0.5-1.15,0);
    \coordinate (A) at (0.5*8/1.75+1.25-1,0);
    \coordinate (B) at (0.5*8/1.75+1.25-1,0.65625);
    \draw[-latex] (0,0) -- (4.25,0) node[right] {$\alpha$}; 
    \draw[-latex] (0.5,-1.5) -- (0.5,2.75) node[below left] {$\omega\left(A\right)$}; 
    

    \draw[very thick,domain=0.5:4.25,smooth,variable=\x,blue] plot({\x},{sqrt(1/4+(\x-0.5)*(\x-0.5))/2*1.75-1/4*8+1});
    \draw[domain=0.5:4.25,smooth,variable=\x,red] (B) plot({\x},{1/2*1.75*(\x-0.5)-1/4*8+1});
    \draw[domain=0.5:4.25,smooth,variable=\x,dashed] plot({\x},{1/2*1.75*(\x-0.5)-1/2*25/8+1});
    \node[inner sep=0.05cm,circle,fill=blue] at (0.5,-0.5*25/8+1) {};
    \node[blue,scale=1.25] at (-0.2,-0.5*25/8+0.5+0.5) {$-\frac{\sigma^2}{2}d_2$};
    \pic [draw,red,thick, ->,angle eccentricity=1.5] {angle = A--O--B};
    \node[red,scale=1] at (0.5*8/1.75+0.5,-0.45) {$\tan(\frac{d_1-d_2}{2\sqrt{2}})$};

    \draw[domain=0.5:4.25,smooth,variable=\x,dashed] plot({\x},{1/2*1.75*(\x-0.5)-1/2*39/8+1});
    \end{tikzpicture}
    \caption{Behavior of $\omega(A)$ for increasing values of $\alpha\geq0$ for the system in Example \ref{ex:numAbscissa}. The dashed lines are the lower and upper bound in~Theorem~\ref{prop:stab}.\vspace{-0.25cm}}
    \label{fig:numAbscissa}
\end{figure}
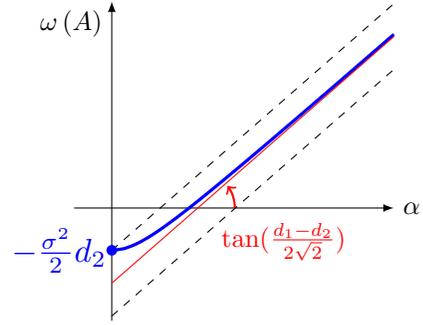
\end{example}

The following theorem generalizes the results of the previous example to the $n$-dimensional case.
\begin{theorem}{\em\bfseries(Relation between $\omega(A)$ and size of $S$)}\label{prop:stab}
Consider the parametrization of $A$ given in Theorem \ref{thm-param} and let $S=\alpha \bar{S}$, where $\alpha\in\real$, $\bar{S}\in\mathcal{S}$, $\bar{S}\ne 0$.  
It holds
\begin{equation}\label{eq:ineq}
\lambda_{\min}(P) + \alpha\lambda_{\max}(M) \leq \omega(A) \leq  \lambda_{\max}(P) + \alpha\lambda_{\max}(M),
\end{equation}
where 
\begin{align*}
    P\! :=\! -\frac{1}{4}\left( \Sigma_w\Sigma^{-1} + \Sigma^{-1} \Sigma_w \right),\ M\! :=\!\frac{1}{2}\left(\bar{S}\Sigma^{-1} - \Sigma^{-1}\bar{S} \right).
\end{align*}
\end{theorem}\bigskip
\begin{proof}
First, note that $P$ and $M$ are symmetric, and
$
    (A + A^\top)/2
    = P + \alpha M .
$
The bounds in \eqref{eq:ineq} now follows by Weyl's inequality \cite[Thm 4.3.1]{horn2012matrix}, which states that for symmetric matrices $X$, $Y$, $\lambda_{\max}(X+Y) \geq \lambda_{\min}(X) + \lambda_{\max}(Y)$, $\lambda_{\max}(X+Y) \leq  \lambda_{\max}(X) + \lambda_{\max}(Y)$.
\end{proof}

From the bounds in Theorem \ref{prop:stab}, it follows that $\omega(A)$ grows linearly with the size of $S$, as quantified by the parameter $\alpha$, if $\lambda_{\max}(M)>0$. The latter condition is satisfied if and only if $M\ne 0$. In fact, the skew-symmetry of $\bar{S}$ yields $\tr(M)=0$, which implies that $M$ is either the zero matrix or has at least a strictly positive eigenvalue. Fig.~\ref{fig:numAbscissa} shows the behavior of $\omega(A)$ together with the bounds in Theorem \ref{prop:stab} as a function of $\alpha$ for the 2-dimensional system of~Example~\ref{ex:numAbscissa}.

To conclude, we note that the lower bound in Theorem \ref{prop:stab} yields a sufficient condition on $S$ to guarantee excitability. Namely, if $M\ne 0$ and $\alpha > -\lambda_{\min}(P)/\lambda_{\max}(M)$, then $\omega(A)>0$ and the system is excitable.

\section{$S$ and frequency-domain behavior of \eqref{eq:model}}\label{sec:frequency}

In this section, we examine the effect of $S$ on the frequency-domain properties of \eqref{eq:model}. Specifically, we analyze how the size of $S$ affects the energy distribution of the system across different frequencies, as determined by its~power~spectral~density \cite[Chap. 10]{lindquist2015linear}:
\begin{equation*}
    \Phi\left(\imath \omega \right) =  \left( \imath \omega I - A \right)^{-1}\Sigma_w \left( \imath \omega I - A \right)^{-*}, \ \ \omega\in\mathbb{R}.
\end{equation*}
We observe that $S$ has no effect on the overall energy of the system since it holds \cite[Chap. 5, Sec. 2]{astrom1970stochastic}
\begin{equation}\label{eq:Sigma}
    \Sigma = \frac{1}{2\pi}\int_{-\infty}^{+\infty}\Phi\left(\imath \omega \right)\mathrm{d}\omega,
\end{equation}
and the steady-state covariance of the system $\Sigma$ is unaffected by $S$. 
However, $S$ can in principle redistribute the energy across different frequency bands. 
This is indeed the case for two-dimensional systems, where $S$ shifts the energy from low to high frequencies, as we illustrate in the next example. 

\begin{example}{\em\bfseries(2-dimensional~system, cont'd)} Consider the system of Example \ref{ex:stab}. The trace of the power spectral density of the system is given by
\begin{equation}\label{eq:freqRes}
    \tr\left[\Phi\left(\imath \omega \right)\right] =\sigma^2 \frac{2\omega^2 + \zeta(\alpha^2)}{\left[ \omega^2 - d_1d_2\zeta(\alpha^2)\right]^2 + \omega^2\frac{\sigma^4}{4}(d_1+d_2)^2},
\end{equation} 
with $\zeta(\alpha^2) := \alpha^2 + {\sigma^4}/{4}$. From \eqref{eq:freqRes}, it follows that for any fixed $\omega\in\mathbb{R}$, $\tr\left[\Phi\left(\imath \omega \right)\right] \to 0$ as $\alpha \to \infty$, which means that as $\alpha$ grows, the energy of the system moves from low to high frequencies. To better understand how $\alpha$ modifies the shape of $\tr\left[\Phi\left(\imath \omega \right)\right]$, we note that \eqref{eq:freqRes} can be viewed as the modulus of the frequency response of a fictitious second-order scalar system featuring a zero at $-\sqrt{\zeta(\alpha^2)/2}$. As standard practice, we can investigate the emergence of resonance phenomena in the system by neglecting the effect of the zero. That is, we can approximate \eqref{eq:freqRes} as
\begin{equation}\label{eq:freqResApp}
    \tr\big[\tilde{\Phi}(\imath \omega)\big] = \frac{\sigma^2}{\left[\omega^2 - d_1d_2\zeta(\alpha^2)\right]^2 + \omega^2\frac{\sigma^4}{4}(d_1+d_2)^2}.
\end{equation} 
Fig.~\ref{fig:ex3} shows a comparison between \eqref{eq:freqRes} and \eqref{eq:freqResApp}. Using the approximation in \eqref{eq:freqResApp} and restricting the analysis to nonnegative frequencies, it can be shown, {after some calculations}, that for $\alpha > \alpha_{\mathrm{th}}:=\sigma^2\sqrt{d_1^2+d_2^2}/\left(2\sqrt{2d_1d_2}\right)$, $\tr\left[\Phi(\imath \omega)\right]$ exhibits a resonance peak at frequency $\omega_r = \sqrt{8d_1d_2\alpha^2-\sigma^4\left(d_1^2+d_2^2\right)}/\left(2\sqrt{2}\right)$. 
We conclude that as $\alpha$ (that is, the size of $S$) increases the system tends to amplify input signals at one particular frequency, and this frequency grows linearly with $\alpha$. \exampleend

\begin{figure}
    \centering
    \includegraphics[width=0.8\linewidth]{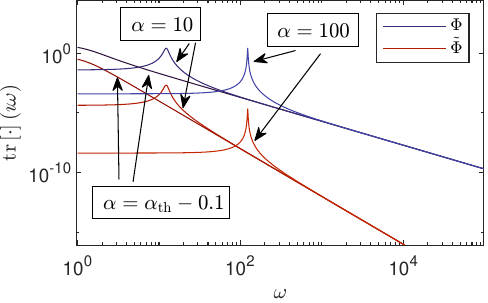}
    \caption{Plot in a logarithmic scale of the trace of the frequency response for the $\tr[\Phi(\imath \omega)]$ and its approximate version $\tr[\tilde{\Phi}(\imath \omega)]$ for different values of $\alpha$, in particular $\alpha = \alpha_{\mathrm{th}} - 0.1, 10, 100$, and $d_1 > d_2 > 0$ have been generated random between $0$ and $10$. $\alpha_{\mathrm{th}}$ denotes the threshold at which the system exhibits resonance.\vspace{-0.25cm}}
    \label{fig:ex3}
\end{figure}
\end{example}

The observations made in the previous example can be generalized as follows.

\begin{theorem}{\em\bfseries(High-pass filtering effect of $S$)} \label{thm:stab}
Consider the parametrization of $A$ given in Theorem \ref{thm-param}, with $S=\alpha \bar{S}$, $\alpha\in\real$, $\bar{S}\in\mathcal{S}$, $\bar{S}\ne 0$. Let
     $$\Phi\left(\imath\omega,\alpha\right) = \left( \imath \omega I - A \right)^{-1}\Sigma_w\left( \imath \omega I - A \right)^{-*}.$$ The following hold:
    \begin{enumerate}
        \item $\exists\, \bar{\omega}\!>\!0$ s.t.~$\tr\left[\Phi\left(\imath\omega,0\right)\right]\! >\! \tr\left[\Phi\left(\imath\omega,\infty\right)\right]$, $\forall\omega$, $|\omega|\!\leq\! \bar{\omega}$;\label{stat:1}
        \item If $\bar{S}$ is full rank,  $\tr\left[\Phi\left(\imath\omega,\alpha\right) \right]\to 0$ as $\alpha\to\infty$,~$\forall \omega\! \in\mathbb{R}$.\label{stat:2}
    \end{enumerate}
\end{theorem}\smallskip
\begin{proof}
Let us define $F(\imath\omega,\alpha) := \left( \imath \omega I - A_{\alpha} \right)^*\Sigma_w^{-1}\left( \imath \omega I - A_{\alpha} \right)$ so that $\Phi\left(\imath \omega,\alpha \right) = F(\imath\omega,\alpha)^{-1}$. Then, it holds
    \begin{equation}\label{eq:princEq}
        F(\imath\omega,\alpha) = \frac{\Sigma^{-1}\Sigma_w\Sigma^{-1}}{4} + \widetilde{F}\left(\imath\omega,\alpha\right),
    \end{equation}
    where 
    \begin{equation}\label{eq:Ftilde}
    \begin{aligned}
    \widetilde{F}\left(\imath\omega,\alpha\right) :=\,& \omega^2 \Sigma_w^{-1} +\omega \alpha \imath\left( \Sigma_w^{-1}\bar{S}\Sigma^{-1} + \Sigma^{-1}\bar{S}\Sigma_w^{-1}\right)\\
    &+\alpha^2\Sigma^{-1}\bar{S}^\top\Sigma_w^{-1}\bar{S}\Sigma^{-1}.
    \end{aligned}
    \end{equation}
    By setting $\omega = 0$, \eqref{eq:princEq} becomes
    \begin{align}\label{eq:changeBas}
    F\left(0,\alpha\right) &= \frac{\Sigma^{-1}\Sigma_w\Sigma^{-1}}{4} + \alpha^2\Sigma^{-1}\bar{S}^\top\Sigma_w^{-1}\bar{S}\Sigma^{-1}.
    \end{align}
    Note that $\Sigma^{-1}\bar{S}^\top\Sigma_w^{-1}\bar{S}\Sigma^{-1}\succeq 0$ since $\Sigma_w,\Sigma \succ 0$. This implies that, for all $\alpha> 0$, $F\left(0,\alpha\right) \succeq F\left(0,0\right)\succ 0$, or, equivalently, $F\left(0,0\right)^{-1} \succeq F\left(0,\alpha\right)^{-1}\succ 0$. In addition, if $\alpha>0$, there is at least one eigenvalue of $F\left(0,0\right)^{-1}$ that is not an eigenvalue of $F\left(0,\alpha\right)^{-1}$.\footnote{Indeed, assume by contradiction that $F\left(0,0\right)^{-1}$ and $F\left(0,\alpha\right)^{-1}$ have the same eigenvalues. This implies that $F\left(0,0\right)$ and $F\left(0,\alpha\right)$ have the same eigenvalues, and therefore the same trace. However, from \eqref{eq:changeBas}, if $\alpha>0$, $\tr[F\left(0,\alpha\right)]> \tr[F\left(0,0\right)]$ since $\bar{S}\ne 0$.} Therefore, for all $\alpha>0$,
    $
    \tr[F\left(0,0\right)^{-1}] > \tr[F\left(0,\alpha\right)^{-1}].
    $
    Point \ref{stat:1}) now follows from the latter inequality and the continuity of $\Phi\left(\imath\omega,\alpha\right)$ with respect to $\omega$.

   To prove point \ref{stat:2}), we introduce the following notation:
\begin{align*}
    &\tilde{a}:= \lambda_{\min}\left(\Sigma^{-1}\Sigma_w\Sigma^{-1}\right), \ \ \tilde{b} := \lambda_{\min}\left(\Sigma_w^{-1}\right),\\
    &\tilde{c} := \lambda_{\min}\left(\Sigma^{-1}\bar{S}^\top\Sigma_w^{-1}\bar{S}\Sigma^{-1} \right),\ \ \tilde{d} := \lVert \Sigma^{-1}\bar{S}\Sigma_w^{-1} \rVert,
\end{align*}
    where $\tilde{a}, \tilde{b}, \tilde{c}, \tilde{d}> 0$ from the the positive definiteness of $\Sigma$, $\Sigma_w$, and the full rank condition of $\bar{S}$. The following hold:
    \begin{align}
    \Sigma^{-1}\Sigma_w\Sigma^{-1} &\succeq \tilde{a}I, \label{eq:a}\\
    \Sigma_w^{-1} &\succeq \tilde{b}I,\label{eq:b}\\
    \Sigma^{-1}\bar{S}^\top\Sigma_w^{-1}\bar{S}\Sigma^{-1} &\succeq \tilde{c}I.\label{eq:c}
    \end{align}
    It remains to bound the term $\omega \alpha \imath\left( \Sigma_w^{-1}\bar{S}\Sigma^{-1} + \Sigma^{-1}\bar{S}\Sigma_w^{-1}\right)$ in \eqref{eq:Ftilde}. First, we see that $\Sigma_w^{-1}\bar{S}\Sigma^{-1} + \Sigma^{-1}\bar{S}\Sigma_w^{-1}$ is skew-symmetric so that $\imath\left( \Sigma_w^{-1}\bar{S}\Sigma^{-1} + \Sigma^{-1}\bar{S}\Sigma_w^{-1}\right)$ is an Hermitian matrix. 
    So it holds
    \begin{equation}\label{eq:bound1Skew}
        \begin{split}&\imath\left( \Sigma_w^{-1}\bar{S}\Sigma^{-1} + \Sigma^{-1}\bar{S}\Sigma_w^{-1}\right) \preceq 2 \tilde{d}I.
        \end{split}
    \end{equation}
    We also know that the eigenvalues of $\imath\left( \Sigma_w^{-1}\bar{S}\Sigma^{-1} + \Sigma^{-1}\bar{S}\Sigma_w^{-1}\right)$ are symmetric with respect to zero. Thus, \eqref{eq:bound1Skew} yields the bounds
    $-2\tilde{d}I \preceq \imath\left( \Sigma_w^{-1}\bar{S}\Sigma^{-1} + \Sigma^{-1}\bar{S}\Sigma_w^{-1}\right) \preceq 2 \tilde{d}I$, which in turn implies
 \begin{equation}\label{eq:finalBound}
         \omega\imath\left( \Sigma_w^{-1}\bar{S}\Sigma^{-1} + \Sigma^{-1}\bar{S}\Sigma_w^{-1}\right) \succeq -2|\omega| \tilde{d}I.
     \end{equation}   
      By applying \eqref{eq:b}, \eqref{eq:c} and \eqref{eq:finalBound} to $\widetilde{F}\left(\imath\omega,\alpha\right)$ in \eqref{eq:princEq}, we obtain
    \begin{align}\label{eq:ineqFTilde}
    \widetilde{F}\left(\imath\omega,S\right) & \succeq \left[\omega^2\tilde{b} -2|\omega|\alpha \tilde{d}+\alpha^2\tilde{c}\right]I \notag \\
    &= \tilde{b}\underbrace{\left[\omega^2 -2|\omega|\alpha{\frac{\tilde{d}}{\tilde{b}}}+\alpha^2{\frac{\tilde{c}}{\tilde{b}}}\right]}_{=: f\left( \imath\omega,\alpha\right)}I.
    \end{align}
    Let $c:={\tilde{c}}/{\tilde{b}}$, $d := {\tilde{d}}/{\tilde{b}}$ and fix any $\bar{\omega} > 0$. 
     Then, for all $|\omega| \leq \bar{\omega}$ and $\alpha \geq \bar{\alpha}:={4d\bar{\omega}}/{c}$  we~have~that:
\begin{align*}
    &f\left( \imath\omega,\alpha\right) = \omega^2 - 2|\omega |\alpha d + \alpha^2 c\\
    & \geq-2|\omega|\alpha d + \alpha^2c=\alpha^2\left[c - 2 \frac{|\omega|}{\alpha}d\right]\geq \alpha^2\left[c-\frac{c}{2}\right] = \alpha^2\frac{c}{2},
\end{align*} 
where the last inequality follows from ${|\omega|}/{\alpha} \leq {\bar{\omega}}/{\bar{\alpha}}= c/(4d)$.
From the latter bound and \eqref{eq:ineqFTilde},
\begin{equation}\label{eq:ineqFTildeFinal}
    \widetilde{F}\left(\imath\omega,\alpha\right) \succeq \tilde{b}\,\alpha^2\frac{c}{2}I.
\end{equation}
By substituting \eqref{eq:ineqFTildeFinal} and \eqref{eq:a} in  \eqref{eq:princEq} we get
\begin{equation*}
    F\left(\imath\omega,\alpha\right) \succeq \left[\frac{\tilde{a}}{4} + \tilde{b}\alpha^2\frac{c}{2}\right]I \quad \text{ if } |\omega| \leq \bar{\omega}.
\end{equation*}
Point \ref{stat:2}) now follows from the latter inequality and the fact that $\bar{\omega}$ can be chosen arbitrarily.  
\end{proof}

Loosely speaking, Theorem \ref{thm:stab} says that when $S$ is large the system attenuates inputs at low frequencies. In particular, if $S$ has full rank and sufficiently large size, low-frequency inputs are completely blocked by the system. Notably, since the overall spectral energy is preserved as $S$ varies (cf.~Eq.~\eqref{eq:Sigma}), Theorem \ref{thm:stab} implies that as $S$ grows the spectral energy of the system moves from low to high frequencies.

\section{Concluding remarks}\label{sec:conclusions}

In this paper, we analyze the properties of stochastic linear systems featuring a prescribed steady-state covariance. These systems are used to model dynamic brain networks generating a fixed resting-state functional connectivity. We show that the state matrix of these systems can be parametrized by a skew-symmetric matrix $S$, and we study the role of $S$ in the stability, excitability, and frequency-domain behavior of the system. By means of theoretical and numerical results, we find that as the entries of $S$ grow the system becomes more stable, excitable, and similar to~an~high-pass~filter.

There are several intriguing directions of future research. In particular, it would be interesting to examine in depth the statistical interpretation of $S$ of Remark \ref{rmk:statistical} and analyze how $S$ affects other system properties, e.g.,~controllability. 

\bibliographystyle{IEEEtran}
\bibliography{sample}

\end{document}